\documentclass[11pt, a4paper]{article}

\usepackage{parskip}
\usepackage{amsmath}
\usepackage{amssymb}
\usepackage{xspace}
\usepackage[nodate]{datetime}
\usepackage{tikz}

\usepackage[hyphens]{url}
\usepackage{bold-extra}
\usepackage{bm}

\usepackage{algorithm}
\usepackage[noend]{algorithmic}

\usepackage{enumitem}
\usepackage{mathdots}

\newtheorem{xdefinition}{Definition}
\newtheorem{xobservation}{Observation}
\newtheorem{xtheorem}{Theorem}
\newtheorem{xlemma}{Lemma}
\newtheorem{xproposition}{Proposition}
\newtheorem{xcorollary}{Corollary}
\newenvironment{definition}{\begin{xdefinition}\rm}%
{\hspace*{\fill}\raisebox{-1pt}{$\blacksquare$}\end{xdefinition}}
{\hspace*{\fill}\raisebox{-1pt}{\boldmath$\Box$}\end{xobservation}}
\newenvironment{theorem}{\begin{xtheorem}\rm}{\end{xtheorem}}
\newenvironment{lemma}{\begin{xlemma}\rm}{\end{xlemma}}

\newenvironment{corollary}{\begin{xcorollary}\rm}{\end{xcorollary}}
\newenvironment{proof}{\begin{trivlist}\item[]{\bf Proof }}%
{\hspace*{\fill}\raisebox{-1pt}{\boldmath$\Box$}\end{trivlist}}

\newcommand{\OPT}{\ensuremath{\operatorname{\textsc{Opt}}}\xspace}

\newcommand{\eps}{\varepsilon}
\newcommand{\abs}[1]{\left| #1\right|}
\newcommand{\FTP}{\ensuremath{\operatorname{\textsc{FtP}}}\xspace}
\newcommand{\GFTP}{\ensuremath{\operatorname{\textsc{GFtP}}}\xspace}

\newcommand{\ZZ}{\ensuremath{\mathbb{Z}}\xspace}
\newcommand{\RR}{\mathbb{R}\xspace}
\newcommand{\PREDW}{\ensuremath{\hat{w}}\xspace}
\newcommand{\TRUEW}{\ensuremath{w}\xspace}
\newcommand{\PREDPERMW}{\PREDW\xspace}
\newcommand{\TRUEPERMW}{\TRUEW\xspace}
\newcommand{\WMSTI}{\ensuremath{(G,\PREDW, \TRUEW)}\xspace}
\newcommand{\ALG}{\ensuremath{\operatorname{\textsc{Alg}}}\xspace}
\newcommand{\BLAME}{\ensuremath{\operatorname{\textsc{Blame}}}\xspace}
\newcommand{\ENEXT}{\ensuremath{e_{\text{next}}}\xspace}
\newcommand{\EFUTURE}{\ensuremath{e_{\text{future}}}\xspace}
\newcommand{\EMAX}{\ensuremath{e_{\operatorname{max}}}\xspace}
\newcommand{\EOPT}{\ensuremath{e_{\OPT}}\xspace}
\newcommand{\BETA}{\ensuremath{\beta}\xspace}
\newcommand{\ROR}{\ensuremath{\operatorname{\textsc{ror}}}\xspace}
\newcommand{\CR}{\ensuremath{\operatorname{\textsc{cr}}}\xspace}
\newcommand{\PREDCOST}{\ensuremath{\hat{c}}\xspace}
\newcommand{\TRUECOST}{\ensuremath{c}\xspace}
\newcommand{\ETA}{\ensuremath{\eta}\xspace}
\newcommand{\PHI}{\ensuremath{\varphi}\xspace}
\newcommand{\ADV}{\ensuremath{\operatorname{\textsc{Adv}}}\xspace}
\newcommand{\ERR}{\ensuremath{\operatorname{\textsc{Err}}}\xspace}

\title{Online Minimum Spanning Trees \\
with Weight Predictions\,\thanks{Supported
in part by the Independent Research Fund Denmark, Natural Sciences,
grant DFF-0135-00018B
and
in part by the Innovation Fund Denmark,
grant 9142-00001B, Digital Research Centre Denmark,
project P40: Online Algorithms with Predictions.}}

\author{
        \begin{tabular}{l@{\hspace{2em}}l}
          Magnus Berg & Joan Boyar \\[1ex]
          Lene M. Favrholdt & Kim S. Larsen \\[2ex]
        \end{tabular} \\
        University of Southern Denmark \\[2ex]
        \texttt{\{magbp,joan,lenem,kslarsen\}@imada.sdu.dk} %\\[1ex]
}

\begin{document}
\maketitle
\begin{abstract}
  We consider the minimum spanning tree problem with predictions,
  using the weight-arrival model, i.e., the graph is given, together
  with predictions for the weights of all edges.  Then the actual
  weights arrive one at a time and an irrevocable decision must be
  made regarding whether or not the edge should be included into the
  spanning tree.
  In order to assess the quality of our algorithms, we
  define an appropriate error measure and analyze the performance of
  the algorithms as a function of the error.
  We prove that, according to competitive analysis, the simplest
  algorithm, Follow-the-Predictions, is optimal.  However,
  intuitively, one should be able to do better, and we present a
  greedy variant of Follow-the-Predictions.  In analyzing that
  algorithm, we believe we present the first random order analysis of
  a non-trivial online algorithm with predictions, by which we obtain
  an algorithmic separation.  This may be useful for distinguishing
  between algorithms for other problems when Follow-the-Predictions is
  optimal according to competitive analysis.

\end{abstract}

\section{Introduction}

The \emph{Minimum Spanning Tree} (MST) problem is one of the classical graph
algorithms problems, where one must select edges from a weighted graph
such that these constitute a spanning tree of minimal weight. We
consider an online version of this problem in the relatively new
context of predictions, a direction that emerged following the
successes of machine learning that has provided more accessible and
reliable predictions.

In the area of online algorithms, we consider problems, many of which
have offline counterparts, where input is presented to an algorithm in
a piece-wise fashion (often referred to as \emph{requests}), and
irrevocable decisions must be made when each item is presented. The
quality of an online algorithm is often assessed using competitive
analysis, which essentially focuses on the worst-case ratio of the cost of the
online algorithm to the cost of an optimal, offline algorithm, \OPT.

When considering graph problems, various models,
inspired by different application scenarios, exist.
In the vertex-arrival model, the requests are the vertices of the graph,
arriving together with the subset of its incident edges that connect to
vertices that have already arrived.
In the edge-arrival model, requests are the edges, identified by their two
endpoints.
For weighted graphs, there is also the \emph{weight-arrival model}, where the
graph is known, and the weights arrive online.
In the vertex-arrival and edge-arrival models, there is only one possible
online algorithm, the one that accepts every edge that does not create a
cycle, since otherwise the algorithm's output might not span the entire graph.
Even in the weight-arrival
model, no deterministic algorithm for online MST can be competitive~\cite{K16}.
This makes the problem hard, but interesting in the context of advice or predictions.

Partially in an attempt to measure how much information about the future is
needed for various online problems, online algorithms with advice
were introduced~\cite{HKK10,EFKR11,DKP09,BFKLM17j}.
In the model used most often, it is an information-theoretical game
of how few oracle-produced bits in total are needed to obtain a particular
competitive ratio or optimality. Obviously, the connection here is that
oracle-based
advice can be considered infallible predictions. The MST problem
has been considered by Bianchi et al.\ in this model~\cite{BBBKP18}.
They obtain results for various arrival models and restricted graph classes, including the
weight-arrival model, but with only two different weights allowed.

The seminal paper by Lykouris and Vassilvitskii~\cite{LV21},
introducing machine-learned advice, which is now more often referred
to as predictions, has inspired rapidly growing~\cite{ALPS} efforts in
the area~\cite{MV22}.  In this context, ideally we want algorithms
to use the predictions and perform optimally when predictions
are correct (referred to as \emph{consistency}), perform as well as
a good online algorithm when predictions are all wrong
(\emph{robustness}), and degrade gracefully from one to the other as
the predictions become increasingly erroneous (\emph{smoothness}).
The ideal situation described above can of course often not be
reached, so one proves upper and lower bounds, as is customary in the
field.  Discussing smoothness requires a definition of error. This is
problem-dependent and requires some thought. We want to distinguish
between good and bad algorithms, and defining error measures that
exaggerate or underestimate the importance of errors leads to
unreliable results.

For the online MST problem with predictions, there are
natural error measures.
We arrive at an error measure, defined as the sum of differences
between the predicted and actual values of the $n-1$ edges (the number of
edges in a spanning tree) with the
largest discrepancies; a measure with desirable properties.

We focus on the MST problem with predictions in the weight-arrival
model.  Our first somewhat surprising result is that with this error
measure (or any of some reasonable alternatives), competitive
analysis~\cite{ST85,KMRS88} cannot distinguish between different,
correct algorithms.  This means that the most na\"{\i}ve algorithm,
Follow-the-Predictions (\FTP), is optimal under this measure, with a
competitive ratio of $1+2\eps$, where $\eps$ is the error, normalized
by the value of \OPT.  Of course, this also means that the perhaps
more reasonable algorithm, we call Greedy Follow-the-Predictions
(\GFTP), that switches to another edge when a revealed actual weight
matches or does better than the predicted weight of an edge it could
replace, is indistinguishable from \FTP under competitive analysis.

In online algorithms, there are other performance
measures one can turn to when competitive analysis is
insufficient, as discussed in~\cite{DL05,BIL15j,BFL20j}.
One of the most well accepted is Random Order Analysis~\cite{K96},
also called the Random Order Model; a chapter in~\cite{R20} discusses some results.
Note that the problem from~\cite{R20} of finding a maximum forest
is not very similar to our problem, since the forest is not required
to be spanning.
The random order analysis technique reduces the power
of the adversary, compared to competitive analysis. In competitive
analysis, the adversary chooses the requests and the order in which
they a presented, while in random order analysis, the adversary
chooses the requests, but those requests are presented to the
algorithm uniformly at random.
Using  random order analysis, we establish a separation between \FTP and \GFTP.
We believe this is the first time random order analysis has been applied in
the context of predictions.

\section{Preliminaries}

Given an online algorithm $\ALG$ for an online minimization problem $\Pi$, and an instance $I$ of $\Pi$, we let $\ALG[I]$ denote $\ALG$'s solution on instance $I$, and $\ALG(I)$ denote the cost of $\ALG[I]$. Then, the \textit{competitive ratio} of $\ALG$ is
\begin{align*}
\CR_{\ALG} = \inf\{ c\ | \ \exists b \colon \forall I \colon \ALG(I)\leq c\OPT(I)-b\}.
\end{align*}
When online algorithms have access to a predictor, a further parameter is introduced into the problem, namely the accuracy of that predictor. Throughout this paper, we let $\ETA$ be the error measure that computes the quality of the predictions, and we let $\eps = \frac{\ETA}{\OPT}$ be the normalized error measure. Our error measure is defined in Definition~\ref{def:mu}. 

Given an online algorithm with predictions, $\ALG$, we express the competitive ratio of $\ALG$ as a function of $\eps$, and evaluate it based on the three criteria: \textit{consistency}, \textit{robustness}, and \textit{smoothness}. Following \cite{LV21}, we define consistency as $\ALG$'s competitive ratio, when the prediction error is $0$. $\ALG$ is \emph{$\gamma$-consistent} if there exists a constant, $\gamma$, such that
%\begin{align*}
$\CR_{\ALG}(0) = \gamma$.
%\end{align*}

An algorithm is \emph{robust} if its competitive ratio is as good as the best online algorithm's (without predictions), independently of how poor the predictions are. In our case, no online algorithm can be competitive, so our algorithms are trivially robust.

As $\eps$ grows, the competitive ratio of $\ALG$ will decay as a
function of $\eps$. For a function, $f$, we say that $\ALG$ is \emph{$f$-smooth}, if
%\begin{align*}
$\CR_{\ALG}(\eps) \leqslant f(\eps)$ for all $\eps$.
%\end{align*}

\subsection{Random Order Analysis}

Given an online algorithm, $\ALG$, for a problem, $\Pi$, and an
instance of $\Pi$ with request sequence $I = \langle i_1,i_2,\ldots,i_n\rangle$, a
permutation $\sigma$ of $I$
is chosen uniformly at random, and $\sigma(I)$ is presented to
$\ALG$.
The \textit{random order ratio} of $\ALG$ is defined as
\begin{align*}
\ROR_{\ALG} = \inf\{ c \ | \ \exists b \colon \forall I \colon \mathbb{E}_{\sigma}[\ALG(\sigma(I))] \leq c\OPT(I)-b\},
\end{align*}
As with the competitive ratio, we express the random order ratio of algorithms with predictions as a function of $\eps$.

\subsection{Weight-Arrival MST Problem}

The offline MST problem is a thoroughly studied problem, for which efficient optimal algorithms are known. Given a graph $G = (V,E,w)$, the task is to find a spanning tree $T$ for $G$ that minimizes the objective function
\begin{align*}
c(T) = \sum_{e \in E(T)} w(e).
\end{align*}
For the MST problem in the weight-arrival model (WMST), online algorithms are initially provided with the underlying graph $G = (V,E)$, and then the weights of the edges in $G$ arrive online. At the time the true weight of an edge $e$ arrives, the online algorithm has to irrevocably accept or reject $e$ for its final tree.
We focus on the WMST problem where we assume that an online algorithm has access to predicted weights for all edges in $G$ before the online computation is initiated.

\subsection{Notation and Nomenclature}

We use the notation $\RR^+$ and $\ZZ^+$ to denote the positive real numbers and the positive integers, respectively.
Graphs, in the following, are weighted, simple, connected and undirected, with positive real weights. Given a graph $G$, we set $n = \abs{V(G)}$ and $m = \abs{E(G)}$. For any clarification on graph theory, we refer to \cite{W01}. Further, we define a \textit{WMST-instance} to be a triple $\WMSTI$ consisting of a graph $G$, and two maps $\PREDW \colon E(G) \rightarrow \RR^+$ and $\TRUEW \colon E(G) \rightarrow \RR^+$, defining for each edge $e \in E(G)$, a predicted weight $\PREDW(e)$ and a true weight $\TRUEW(e)$. Given a graph $G$ and a tree $T \subset G$, when writing $T$, we implicitly refer to $E(T)$. Moreover, we let
\begin{align*}
\PREDCOST(T)=\sum_{e \in T} \PREDW(e)\hspace{0.5cm} \text{and} \hspace{0.5cm} \TRUECOST(T)=\sum_{e \in T} \TRUEW(e).
\end{align*}
Given an algorithm with predictions, $\ALG$, for the WMST problem, and a WMST-instance $\WMSTI$, we let $\ALG[\PREDW(E(G)),\TRUEW(E(G))]$ denote the tree that $\ALG$ outputs. When $G$ is clear from the context, we write $\ALG[\PREDW,\TRUEW]$ and let $\ALG(\PREDW,\TRUEW)$ denote the cost of $\ALG[\PREDW,\TRUEW]$. We let $\OPT[\TRUEW]$ be an optimal MST of $G$, and $\OPT[\PREDW]$ be an optimal MST of $G$ using the predicted weights $\PREDW$.

\subsection{Pictorial Representations of WMST-Instances}
Given a WMST-instance $\WMSTI$, when representing $G$ pictorially, we denote the predicted and true weights of an edge $e \in E(G)$ by $\PREDW(e) \rightarrow \TRUEW(e)$. Thus, the WMST-instance $\WMSTI$ given by $V(G) = \{u_1,u_2,u_3\}$, $E(G) = \{(u_1,u_2),(u_2,u_3),(u_1,u_3)\}$, 
\begin{align*}
\PREDW(e) = \begin{cases}
3, &\mbox{if $e = (u_2,u_3)$,} \\
2, &\mbox{otherwise}
\end{cases} \hspace{0.5cm} \text{and} \hspace{0.5cm} \TRUEW(e) = \begin{cases}
2, &\mbox{if $e = (u_1,u_3)$,} \\
1, &\mbox{otherwise}
\end{cases}
\end{align*}
may be pictorially represented by
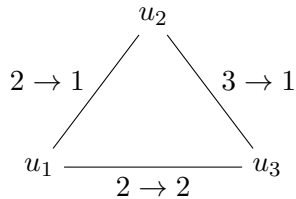
\begin{center}
\begin{tikzpicture}
\node (u1) at (0,0) {$u_1$};
\node (u2) at (1.5,2) {$u_2$};
\node (u3) at (3,0) {$u_3$};

\draw (u1) -- node[left, pos = 0.6] {$2 \rightarrow 1\ $} (u2);
\draw (u2) -- node[right, pos = 0.4] {$\ 3 \rightarrow 1$} (u3);
\draw (u3) -- node[below] {$2 \rightarrow 2$} (u1);
\end{tikzpicture}
\end{center}

\subsection{Measure Comparison}

When selecting the error measure for evaluating the quality of a specific prediction scheme, one has to ensure that the error measure satisfies certain desirable properties, and that it picks up salient features of the specific problem. In our case, a natural first idea is, given a WMST-instance $\WMSTI$, to define
\begin{align*}
\ETA_1(\PREDW,\TRUEW) = \sum_{e \in E(G)} \abs{\PREDW(e) - \TRUEW(e)}.
\end{align*}
This choice, however, suffers the flaw that it cannot separate our algorithms, and, by the definition of $\ETA_1$, dense graphs will have potential for unreasonably large prediction errors. In \cite{PSK18}, Kumar, Purohit, and Svitkina suggest the same measure for Non-Clairvoyant Scheduling, where they sum over the prediction error of each job size. Based on their work, Im et al.\ \cite{IKQP21} propose an alternative measure for the same problem, having more desirable properties, and were sensitive to further important problem-specific parameters. In particular, Im et al.\ suggest that error measures be \emph{monotone} and satisfy a \emph{Lipschitz}-like property defined as follows.

\begin{definition}
Let $\WMSTI$ be a WMST-instance. Then, an error measure, $\ERR(\PREDW(E(G)),\TRUEW(E(G)))$, is said to be \emph{monotone} if, for all subgraphs $G' \subset~G$, 
\begin{align*}
\ERR(\PREDW(E(G) \setminus E(G')) \cup \TRUEW(E(G')) , \TRUEW(E(G))) \leqslant \ERR(\PREDW(E(G)),\TRUEW(E(G))).
\end{align*}
\end{definition}

In words, an error measure, $\ERR$, is said to be monotone if the action of correcting a subset of predicted weights to the correct weights does not increase the value of the error.

\begin{definition}
Let $\WMSTI$ be a WMST-instance. Then, an error measure, $\ERR(\PREDW(E(G)),\TRUEW(E(G)))$, is said to be \emph{Lipschitz} if, 
\begin{align*}
\abs{\OPT(\PREDW(E(G))) - \OPT(\TRUEW(E(G)))} \leqslant \ERR(\PREDW(E(G)),\TRUEW(E(G)))
\end{align*}
\end{definition}

Specifically, Im et al.\ suggest a measure, $\ETA_2$. In our setting, we can define $\ETA_2(\PREDW,\TRUEW)$ as
\[\OPT(\{\PREDW(e)\}_{e\in E_o} \cup \{\TRUEW(e)\}_{e \in E_u}) - \OPT(\{\TRUEW(e)\}_{e \in E_o} \cup \{\PREDW(e)\}_{e \in E_u}),
\]
where $E_o = \{e \in E(G) \mid \PREDW(e) > \TRUEW(e)\}$ and $E_u =
E(G) \setminus E_o$. In our setting, this measure also fails to
distinguish algorithms. In particular, no online algorithm can have a
competitive ratio that is a function of $\eps_2$, or even $\ETA_2$. 

\begin{theorem}
For any deterministic online algorithm with predictions, $\ALG$, for the WMST problem, and any function, $f$, there exists a WMST-instance $\WMSTI$ such that
\begin{align*}
\frac{\ALG(\PREDW,\TRUEW)}{\OPT(\TRUEW)} > f(\ETA_2).
\end{align*}
\end{theorem}
\begin{proof}
  For any $k \in \ZZ^+$, consider the WMST-instance $(G_{k,K},\PREDW_{k,K},\TRUEW_{k,K})$, where $K$ depends on the actions of \ALG:
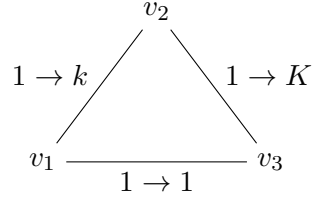
\begin{center}
\begin{tikzpicture}
\node (1) at (0,0) {$v_1$};
\node (2) at (1.5,2) {$v_2$};
\node (3) at (3,0) {$v_3$};

\draw (1) -- node[left, pos = 0.6] {$1 \rightarrow k\ $} (2);
\draw (1) -- node[below] {$1 \rightarrow 1$} (3);
\draw (2) -- node[right, pos = 0.4] {$\ 1 \rightarrow K$} (3);
\end{tikzpicture}
\end{center}
First, the adversary, $\ADV$, reveals $\TRUEW_{k,K}((v_1,v_2)) =
k$. We now have two cases:

\textbf{Case ($\ALG$ accepts $\bm{{(v_1,v_2)}}$):} In this case, $\ADV$ sets $K=1$, and so
\begin{itemize}[label = {-}]
\item $\ALG(\PREDW_{k,1},\TRUEW_{k,1}) = k+1$.
\item $\OPT(\TRUEW_{k,1}) = 2$.
\item $\ETA_2(\PREDW_{k,1},\TRUEW_{k,1}) = 0$.
\end{itemize}
Hence, 
\begin{align*}
\frac{\ALG(\PREDW_{k,1},\TRUEW_{k,1})}{\OPT(\TRUEW_{k,1})} = \frac{k+1}{2}.
\end{align*}
Since $k$ can be arbitrarily large, and since $\ETA_2=0$, this
fraction cannot be bounded by any function of $\ETA_2$.

\noindent \textbf{Case ($\ALG$ rejects $\bm{{(v_1,v_2)}}$):} In this
case, $\ADV$ chooses $K > k$, and so
\begin{itemize}[label = {-}]
\item $\ALG(\PREDW_{k,K},\TRUEW_{k,K}) = K+1$.
\item $\OPT(\TRUEW_{k,K}) =  k + 1$.
\item $\ETA_2(\PREDW_{k,K},\TRUEW_{k,K}) = k - 1$.
\end{itemize}
Hence, 
\begin{align*}
\frac{\ALG(\PREDW_{k,K},\TRUEW_{k,K})}{\OPT(\TRUEW_{k,K})} =
\frac{K+1}{k+1} = \frac{K+1}{\ETA_2+2}.
\end{align*}
Since $K$ may be chosen arbitrarily large, independently of $k$, and
hence of $\ETA_2$,
this fraction cannot be bounded by a function of $\eta_2$.
\end{proof}

We use the following measure, denoted by $\ETA$, selected due to its desirable properties and its ability to distinguish between algorithms under random order analysis.

\begin{definition}\label{def:mu}
Let $\WMSTI$ be any WMST-instance, and let $e_1,e_2,\ldots,e_m$ be any ordering of $E(G)$. Furthermore, let $\{p_i\}_i$ be the sequence where $p_i := \abs{\TRUEW(e_i) - \PREDW(e_i)}$, and let $\{p_{i_j}\}_j$ be the sequence $\{p_i\}_i$, sorted such that $p_{i_1} \geqslant p_{i_2} \geqslant \cdots \geqslant p_{i_m}$. The \emph{error}, $\ETA$, is given by
\begin{align*}
\ETA(\PREDW,\TRUEW) = \sum_{j=1}^{n-1} p_{i_j}.
\end{align*}
When $\WMSTI$ is clear from context, we write $\ETA$ for $\ETA(\PREDW,\TRUEW)$. The \emph{normalized error} is $\eps = \frac{\ETA}{\OPT}$. 
\end{definition}

Note that $n-1$ is the number of edges in a spanning tree.
Thus, the risk of unreasonably large prediction errors for dense graphs as with possible other error measures has been eliminated. This measure also satisfies the monotonicity and Lipschitzness properties from~\cite{IKQP21}. 

\begin{theorem}
$\ETA$ is monotone and Lipschitz.
\end{theorem}
\begin{proof}
\textbf{Towards monotonicity:} Given a WMST-instance $\WMSTI$ and any enumeration $e_1,e_2,\ldots,e_m$ of $E(G)$, we set $p_i := \abs{\TRUEW(e_i) - \PREDW(e_i)}$ obtaining a sequence $\{p_i\}_{i}$ of prediction errors. Now, sort $\{p_i\}_i$ in non-increasing order, to obtain $\{p_{i_j}\}_{j}$. Then,
\begin{align*}
\ETA(\PREDW,\TRUEW) = \sum_{j=1}^{n-1} p_{i_j}.
\end{align*}
Correcting predictions by setting $\PREDW(e_i) := \TRUEW(e_i)$, for some $i \in \{1,2,\ldots,m\}$, cannot make $\ETA$ increase. Indeed, if $p_i = \abs{\PREDW(e_i) - \TRUEW(e_i)}$ did not contribute to $\ETA$ before correcting $\PREDW(e_i)$, then $\ETA$ remains unchanged after the correction. If, on the other hand, $p_i$ contributed to $\ETA$ before, we find that instead of $p_i$, now the $n$th largest prediction error, before correcting $\PREDW(e_i)$, will contribute to $\ETA$ instead. Since $p_n \leqslant p_i$, it follows that $\ETA$ can only either remain unchanged or decrease after the correction. 

\textbf{Towards Lipschitzness:} We show that $\abs{\PREDCOST(T_{\FTP}) - \TRUECOST(T_{\OPT})} \leqslant \ETA$, which is equivalent to
\begin{align}\label{eq:lipschitzness}
\TRUECOST(T_{\OPT}) - \ETA \leqslant \PREDCOST(T_{\FTP}) \leqslant \TRUECOST(T_{\OPT}) + \ETA.
\end{align}
We prove the two inequalities separately. To this end, by the minimality of $\OPT$, observe that
\begin{enumerate}[label = {(\roman*)}]
\item $\PREDCOST(T_{\FTP}) \leqslant \PREDCOST(T_{\OPT})$, and
\item $\TRUECOST(T_{\OPT}) \leqslant \TRUECOST(T_{\FTP})$. 
\end{enumerate}
Moreover, by the observations in Theorem~\ref{thm:cr_ftp_upper_bound}, we find that
\begin{enumerate}[label = {(\alph*)}]
\item $\TRUECOST(T_{\FTP}) \leqslant \PREDCOST(T_{\FTP}) + \ETA$, and
\item $\PREDCOST(T_{\OPT}) - \ETA \leqslant \TRUECOST(T_{\OPT})$. 
\end{enumerate}
Now, by (a) and (ii), it follows that $\TRUECOST(T_{\OPT}) \leqslant \PREDCOST(T_{\FTP}) + \ETA$, which is equivalent to $\TRUECOST(T_{\OPT}) - \ETA \leqslant \PREDCOST(T_{\FTP})$, implying the leftmost inequality in Equation~\eqref{eq:lipschitzness}. Similarly, by (b) and (i), it follows that $\PREDCOST(T_{\FTP}) - \ETA \leqslant \TRUECOST(T_{\OPT})$, which is equivalent to $\PREDCOST(T_{\FTP}) \leqslant \TRUECOST(T_{\OPT}) + \ETA$, which implies the rightmost inequality in Equation~\eqref{eq:lipschitzness}. 
\end{proof}

\section{Optimal Algorithms under Competitive Analysis}

We prove that our two algorithms $\FTP$ and $\GFTP$, defined in
Algorithms~\ref{alg:ftp} and \ref{alg:gftp}, respectively, are
$1$-consistent and ($1+2\eps$)-smooth algorithms and that this is best
possible.

\subsection{Upper Bounds}

First, we focus on the simplest algorithm, called Follow-the-Predictions ($\FTP$), defined in Algorithm~\ref{alg:ftp}. 

\begin{algorithm}[ht!]
\caption{$\FTP$}
\begin{algorithmic}[1]
\STATE \textbf{Input:} A WMST-instance $\WMSTI$
\STATE Let $T$ be a MST of $G$ using $\PREDW$
\WHILE {receiving inputs $(\TRUEW(e_i),e_i)$}
	\IF {$e_i \in T$}
		\STATE Accept $e_i$ \COMMENT{Add $e_i$ to the solution}
	\ENDIF
\ENDWHILE
\end{algorithmic}
\label{alg:ftp}
\end{algorithm}

For brevity, we set $T_{\FTP} = \OPT[\PREDW]$ and $T_{\OPT} = \OPT[\TRUEW]$.

\begin{theorem}\label{thm:cr_ftp_upper_bound}
$\CR_{\FTP}(\eps) \leqslant 1 + 2 \eps$.
\end{theorem}
\begin{proof}
First, note that
\begin{align*}
\TRUECOST(T_{\FTP}) - \PREDCOST(T_{\FTP}) = \sum_{e \in T_{\FTP}} \left( \TRUEW(e) - \PREDW(e) \right) \leqslant \sum_{e \in T_{\FTP}} \abs{\TRUEW(e) - \PREDW(e)} \leqslant \ETA.  
\end{align*}
A similar argument shows that $\PREDCOST(T_{\OPT}) \leqslant \TRUECOST(T_{\OPT}) + \ETA$. By the minimality of $T_{\FTP}$ with respect to $\PREDW$, it follows that
\begin{align}\label{eq:upper_bound_computation}
\TRUECOST(T_{\FTP}) - \ETA \leqslant \PREDCOST(T_{\FTP}) \leqslant \PREDCOST(T_{\OPT}) \leqslant \TRUECOST(T_{\OPT}) + \ETA,
\end{align}
and, therefore, that
%\begin{align*}
$\TRUECOST(T_{\FTP}) \leqslant \TRUECOST(T_{\OPT}) + 2 \ETA.$
%\end{align*}
Since $\TRUECOST(T_{\FTP}) = \FTP(\PREDW)$ and $\TRUECOST(T_{\OPT}) = \OPT(\TRUEW)$, it follows that
\begin{align*}
\frac{\FTP(\PREDW)}{\OPT(\TRUEW)} \leqslant 1 + 2 \eps.
\end{align*}
\end{proof}

We also present a non-trivial algorithm, called Greedy-\FTP (\GFTP)
that starts by producing the tree that \FTP outputs. Whenever the true
weight of an edge, $e$, that is not contained in $\GFTP$'s current
tree is revealed, the algorithm checks whether $e$ can replace an edge in its current tree. It does so by comparing the predicted weights of a subset of edges in its current tree by the newly revealed true weight. We formalize the strategy of \GFTP in Algorithm~\ref{alg:gftp}.

\begin{algorithm}[ht!]
\caption{$\GFTP$}
\begin{algorithmic}[1]
\STATE \textbf{Input:} A WMST-instance $\WMSTI$
\STATE Let $T$ be a MST of $G$ using $\PREDW$
\STATE $U = E(G)$ \COMMENT{$U$ contains the \textit{unseen} edges}
\WHILE {receiving inputs $(\TRUEW(e_i),e_i)$}
	\STATE $U = U \setminus \{e_i\}$ 
	\IF {$e_i \in T$}
		\STATE Accept $e_i$ \COMMENT{Add $e_i$ to the solution} \label{alg-in-accept}
	\ELSE[$e_i \not\in T$]
		\STATE $C$ is the cycle $e_i$ introduces in $T$
		\STATE $C' = U \cap C$ 
		\IF {$C' \neq \emptyset$}
			\STATE $\EMAX = \text{arg}\,\max_{e_j \in C'}\{\PREDW(e_j)\}$
			\IF {$\TRUEW(e_i) \leqslant \PREDW(\EMAX)$} 
				\STATE $T = (T \setminus \{\EMAX\}) \cup \{e_i\}$ \COMMENT{Update $T$}
				\STATE Accept $e_i$ \COMMENT{Add $e_i$
                                  to the solution} \label{alg-out-accept}
			\ENDIF
		\ENDIF
	\ENDIF
\ENDWHILE
\end{algorithmic}
\label{alg:gftp}
\end{algorithm}

Throughout, we set $T_{\GFTP} = \GFTP[\PREDW,\TRUEW]$. Further, we denote by $T$ the tree that $\GFTP$ makes online changes to. Note that initially $T = T_{\FTP}$, and after $\GFTP$ has processed the full input sequence, $T = T_{\GFTP}$. Finally, we denote by $U$ the collection of \emph{unseen} edges in $E(G)$.

\begin{lemma}\label{lem:cftp_compared_to_ftp}
For any WMST-instance $\WMSTI$, $\TRUECOST(T_{\GFTP}) \leqslant \PREDCOST(T_{\FTP}) + \ETA$. 
\end{lemma}
\begin{proof}
There exists a bijection $\PHI \colon T_{\FTP} \rightarrow T_{\GFTP}$, where, for each $e \in T_{\FTP}$, 
\begin{align*}
\PHI(e) = \begin{cases}
e, &\mbox{if $\GFTP$ accepted $e \in T$ in Line~\ref{alg-in-accept} of Algorithm~\ref{alg:gftp}}, \\
e', &\mbox{if $\GFTP$ swapped out $e$ for $e'$ in Line~\ref{alg-out-accept} of Algorithm~\ref{alg:gftp}}. 
\end{cases}
\end{align*}
Clearly, $\varphi$ is a surjection, and since $\abs{T_{\FTP}} = \abs{T_{\GFTP}}$, $\varphi$ is bijective. Hence,
\begin{align*}
\TRUECOST(T_{\GFTP}) - \PREDCOST(T_{\FTP}) = \sum_{e \in T_{\FTP}} \TRUEW(\PHI(e)) - \PREDW(e).
\end{align*}
Given an edge $e \in T_{\FTP}$, if $\PHI(e) = e$, then $\TRUEW(\PHI(e)) \leqslant \PREDW(e) + \abs{\TRUEW(e) - \PREDW(e)}$, and so $\TRUEW(\PHI(e)) - \PREDW(e) \leqslant \abs{\TRUEW(e) - \PREDW(e)}$. If $\PHI(e) \neq e$, then, by Algorithm~\ref{alg:gftp}, $\TRUEW(\PHI(e)) \leqslant \PREDW(e)$. Hence,
\begin{align*}
\TRUECOST(T_{\GFTP}) - \PREDCOST(T_{\FTP}) = \sum_{e \in T_{\FTP}} \TRUEW(\PHI(e)) - \PREDW(e) \leqslant \sum_{e \in T_{\FTP}} \abs{\TRUEW(e) - \PREDW(e)} \leqslant \ETA.
\end{align*}
\end{proof}

\begin{theorem}\label{thm:cr_cftp_upper_bound}
$\CR_{\GFTP}(\eps) \leqslant 1 + 2\eps$.
\end{theorem}
\begin{proof}
Using Lemma~\ref{lem:cftp_compared_to_ftp} and following the proof of Theorem~\ref{thm:cr_ftp_upper_bound} from Equation~\eqref{eq:upper_bound_computation}, the result follows.
\end{proof}

\subsection{Lower Bounds}

We establish that the competitive ratios from Theorems~\ref{thm:cr_ftp_upper_bound} and \ref{thm:cr_cftp_upper_bound} are tight.

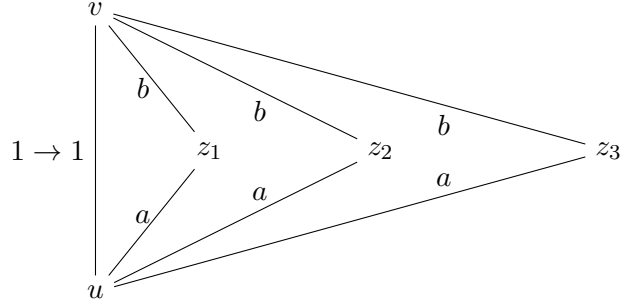
\begin{figure}[ht!]
\centering
\begin{tikzpicture}[scale = 0.75]
%\draw[gray!25] (0,0) grid (8,5);
\node (u1) at (0,0) {$u$};
\node (u2) at (0,5) {$v$};

\draw (u1) -- node[left] {$1 \rightarrow 1$} (u2);

\node (v1) at (2,2.5) {$z_1$};
\node (v2) at (5,2.5) {$z_2$};
\node (v3) at (9,2.5) {$z_3$};

\draw (v1) -- node[above, pos = 0.6] {$a$} (u1);
\draw (v1) -- node[below, pos = 0.6] {$b$} (u2);
\draw (v2) -- node[above, pos = 0.4] {$a$} (u1);
\draw (v2) -- node[below, pos = 0.4] {$b$} (u2);
\draw (v3) -- node[above, pos = 0.3] {$a$} (u1);
\draw (v3) -- node[below, pos = 0.3] {$b$} (u2);
\end{tikzpicture}
\caption{The graph $G_{k,3}$, when $E_{T_{\FTP}} = \{(z_i,v) \mid i = 1,2,3\}$. Here, $a$ denotes the weights $k+1 \rightarrow 1$ and $b$ denotes the weights $k+1 \rightarrow 2k + 1$.}
\label{fig:ftp_lower_bound}
\end{figure}

\begin{theorem}\label{thm:cr_ftp_lower_bound}
For all $r<2$, there exists a WMST-instance $\WMSTI$, such that 
\begin{align*}
\frac{\FTP(\PREDW)}{\OPT(\TRUEW)} \geqslant 1 + r \eps.
\end{align*}
\end{theorem}
\begin{proof}
For $k>1$ and $\ell \in \ZZ^+$, define the WMST-instance $(G_{k,\ell},\PREDW_{k,\ell},\TRUEW_{k,\ell})$, with $G_{k,\ell}=(V,E)$, as (see Figure~\ref{fig:ftp_lower_bound}):
\begin{itemize}
\item $V=\{ u,v\} \cup \{ z_i \mid 1\leqslant i\leqslant \ell\}$,
\item $E=\{ (u,v)\} \cup E_{T_{\FTP}} \cup E_{T_{\OPT}}$, where $E_{T_{\FTP}}= T_{\FTP} \setminus \{(u,v)\}$ and $E_{T_{\OPT}}=(E \setminus \{(u,v)\}) \setminus E_{T_{\FTP}}$,
\item $\PREDW_{k,\ell}((u,v)) = 1$ and $\TRUEW_{k,\ell}((u,v)) = 1$, and
\item for all $e \in E\setminus \{(u,v)\}\colon \PREDW_{k,\ell}(e) = k + 1$.
\end{itemize}
Then, the adversary sets $\TRUEW_{k,\ell}(e) = 2k + 1$ for all $e \in E_{T_{\FTP}}$, and $\TRUEW_{k,\ell}(e) = 1$, for all $e \in E_{T_{\OPT}}$. There is a prediction error of $k$ on each edge, except for $(u,v)$. Hence,
\begin{itemize}
\item $\OPT(\TRUEW_{k,\ell}) = \ell+1$,
\item $\FTP(\PREDW_{k,\ell}) =  \ell(2k+1)+1$,
\item $\ETA = (\ell+1)k$ and $\eps=k$.
\end{itemize}
From this, it follows that
\begin{align*}
  \FTP(\PREDW_{k,\ell}) & = \OPT(\TRUEW_{k,\ell}) + 2k\ell \\
& = \OPT(\TRUEW_{k,\ell}) + 2k\ell + 2k - 2k \\
& = \OPT(\TRUEW_{k,\ell}) + 2k(\ell + 1) - 2k \\
& = \OPT(\TRUEW_{k,\ell}) + 2\eps\OPT(\TRUEW_{k,\ell}) - 2k,
\end{align*}
so
\begin{align*}
\frac{\FTP(\PREDW_{k,\ell})}{\OPT(\TRUEW_{k,\ell})} = 1 + \left(2-\frac{2}{(\ell+1)}\right)\eps.
\end{align*}
For all $r<2$, there exist $\ell \in \ZZ^+$ such that $2-\frac{2}{(\ell+1)}\geqslant r$.
\end{proof}

\begin{corollary}\label{thm:cr_cftp_lower_bound}
For all $r<2$, there exists a WMST-instance $\WMSTI$, such that
\begin{align*}
\frac{\GFTP(\PREDW,\TRUEW)}{\OPT(\TRUEW)} \geqslant 1 + r \eps.
\end{align*}
\end{corollary}
\begin{proof}
  With the same set-up as in Theorem~\ref{thm:cr_ftp_lower_bound}, the adversary now additionally forces $\GFTP$ to pick the same tree as $\FTP$ by revealing the true weights of all the edges in its initial tree, $\OPT[\PREDW]$, before all other edges.
\end{proof}

\begin{corollary}
$\CR_{\FTP}(\eps) = 1 + 2\eps$ and $\CR_{\GFTP}(\eps) = 1 + 2 \eps$.
\end{corollary}
\begin{proof}
This is a direct consequence of Theorems~\ref{thm:cr_ftp_upper_bound}, \ref{thm:cr_cftp_upper_bound}, and \ref{thm:cr_ftp_lower_bound}, and Corollary~\ref{thm:cr_cftp_lower_bound}.
\end{proof}

We establish a general lower bound for deterministic online algorithms with predictions for the WMST problem. The existence of this lower bound shows that, under competitive analysis, any online algorithm with predictions, $\ALG$, which, for any WMST-instance, $\WMSTI$, guarantees that $\ALG(\PREDW,\TRUEW) \leqslant \OPT(\TRUEW) + 2\ETA$ is asymptotically optimal in $\ETA$. Thus, both $\FTP$ and $\GFTP$ are optimal.

\begin{theorem}\label{thm:general_lower_bound}
For the WMST problem with weight predictions, for any algorithm, \ALG, and any $r < 2$, there exists a WMST-instance $\WMSTI$ such that
\begin{align*}
\ALG(\PREDW,\TRUEW) \geqslant \OPT(\TRUEW) + r\ETA. 
\end{align*}
\end{theorem}
\begin{proof}
Fix $k \in \ZZ^+$, such that $k > 1$. For any $\ell \in \ZZ^+$, define $(G_{k,\ell},\PREDW_{k,\ell},\TRUEW_{k,\ell})$, with $G_{k,\ell} = (V,E)$ as follows; see Figure~\ref{fig:lower_bound}:

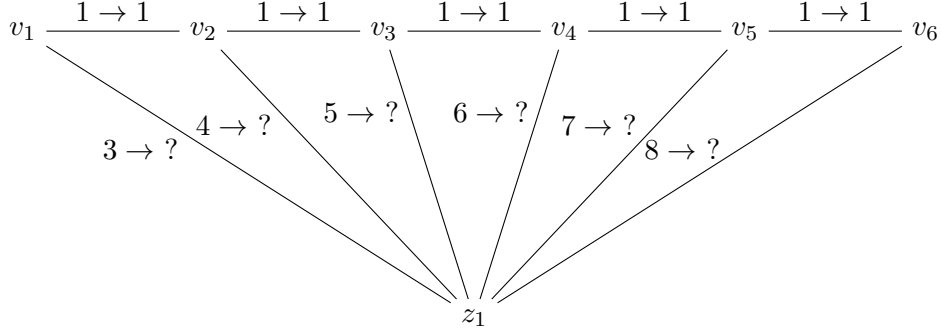
\begin{figure}[ht!]
\centering
\begin{tikzpicture}[scale=0.95]
% Initial 2k vertices
\node (1) at (0,4) {$v_1$};
\node (2) at (2.5,4) {$v_2$};
\node (3) at (5,4) {$v_3$};
\node (4) at (7.5,4) {$v_4$};
\node (5) at (10,4) {$v_5$};
\node (6) at (12.5,4) {$v_6$};

% Edges connecting initial 2k vertices
\draw (1) -- node[above] {$1 \rightarrow 1$} (2);
\draw (2) -- node[above] {$1 \rightarrow 1$} (3);
\draw (3) -- node[above] {$1 \rightarrow 1$} (4);
\draw (4) -- node[above] {$1 \rightarrow 1$} (5);
\draw (5) -- node[above] {$1 \rightarrow 1$} (6);

% Node constituting first 1-vertex component
\node (l1) at (6.25,0) {$z_1$};

% Edges from first 1-vertex component to initial vertices
\draw (l1) -- node[left, pos = 0.6] {$3\rightarrow\ ?\ \ $} (1);
\draw (l1) -- node[left, pos = 0.7] {$4\rightarrow\ ?\ $} (2);
\draw (l1) -- node[left, pos = 0.75] {$5\rightarrow\ ?$} (3);
\draw (l1) -- node[left, pos = 0.75] {$6\rightarrow\ ?$} (4);
\draw (l1) -- node[left, pos = 0.7] {$7\rightarrow\ ?\ $} (5);
\draw (l1) -- node[left, pos = 0.6] {$8\rightarrow\ ?\ $} (6);

\end{tikzpicture}
\caption{Lower bound graph $G_{3,1}$.}
\label{fig:lower_bound}
\end{figure}

\begin{itemize}[label = {$\bullet$}]
\item $V = \{v_1,v_2,\ldots,v_{2k}\} \cup \{z_j \mid 1 \leqslant j \leqslant \ell\}$,
\item $E = I \cup \bigcup_{j=1}^{\ell} E_j$, where 
\begin{itemize}[label = {-}]
\item $I = \bigcup_{i=1}^{2k-1} \{(v_i,v_{i+1})\}$,
\item $E_j = \{(z_j,v_i) \mid 1 \leqslant i \leqslant 2k\}$, for all $j=1,2,\ldots,\ell$,
\end{itemize}
\item $\forall e \in I$, $\PREDW_{k,\ell}(e) = \TRUEW_{k,\ell}(e) = 1$,
\item $\forall (z_j,v_i) \in E_j$, $\PREDW_{k,\ell}((z_j,v_i)) = k + i - 1$; $\ADV$ determines $\TRUEW_{k,\ell}((z_j,v_i))$ as shown in Algorithm~\ref{alg:adv_general_lower_bound}:
\end{itemize} 

\begin{algorithm}[ht!]
\caption{$\ADV$}
\begin{algorithmic}[1]
\FOR {$i=1,2,\ldots,2k-1$}
	\STATE Reveal $\TRUEW_{k,\ell}(v_i,v_{i+1}) = 1$
\ENDFOR
\FOR {$j=1,2,\ldots,\ell$}
	\STATE Set $\TRUEW_{k,\ell}(z_j,v_1) = 2k$
	\FOR {$i=1,2,\ldots,2k-1$}
		\STATE Reveal $\TRUEW_{k,\ell}(z_j,v_i)$
		\IF {$\ALG$ accepts $(z_j,v_i)$}
			\STATE Set $\TRUEW_{k,\ell}(z_j,v_{i+1}) = i$
		\ELSE % {$\ALG$ rejects $(z_j,v_i)$}
			\STATE Set $\TRUEW_{k,\ell}(z_j,v_{i+1}) = 2k+i$
		\ENDIF
	\ENDFOR
	\STATE Reveal $\TRUEW_{k,\ell}(z_j,v_{2k})$
\ENDFOR
\end{algorithmic}
\label{alg:adv_general_lower_bound}
\end{algorithm}

Since all edges in any $E_j$ have weight greater than $1$, $\ALG$ performs best if it accepts the first $2k-1$ edges and then exactly one edge from $E_j$ for each $j$. Since we are proving a lower bound, we assume this.

Observe that for each $e \in \bigcup_{j=1}^{\ell}E_j$, we have that
\begin{align*}
\abs{\PREDW_{k,\ell}(e) - \TRUEW_{k,\ell}(e) } = k. 
\end{align*}
Hence, by definition of $\ETA$, and since $\abs{V(G_{k,\ell})} = 2k+\ell$, it follows that
\begin{align*}
\ETA_{k,\ell} = \ETA(\PREDW_{k,\ell},\TRUEW_{k,\ell}) = (2k + \ell - 1)k = 2k^2 + k\ell - k.
\end{align*}
Now, we analyze the behavior of $\ALG$ and $\OPT$ to asses their cost difference. Since $\ALG$ has to pick exactly one edge from each $E_j$, there are two cases: \\
\textbf{Case: $\ALG$ accepts $\bm{{(z_j,v_i)}}$, for some $\bm{{i \neq 2k}}$.} Then, $\ADV$ sets the true weight of $(z_jv_{i+1})$ to be $\TRUEW_{k,l}((z_j,v_{i+1})) = i$, and so, $\OPT$ accepts $(z_j,v_i)$. The cost difference between $\ALG$'s and $\OPT$'s choice is
\begin{align*}
2k + i - 1 - i = 2k - 1.
\end{align*}
\textbf{Case: $\ALG$ accepts $\bm{{(z_j,v_{2k})}}$.} $\TRUEW_{k,\ell}((z_j,v_{2k})) = 2k + 2k - 1 = 4k-1$. $\OPT$ picks $(z_j,v_1)$, and so, the cost difference between $\ALG$'s and $\OPT$'s choice is
\begin{align*}
4k-1-2k = 2k-1.
\end{align*}
Thus,
\begin{align*}
\ALG(\PREDW_{k,\ell},\TRUEW_{k,\ell}) &\geqslant \OPT(\TRUEW_{k,\ell}) + \ell(2k-1) \\
&= \OPT(\TRUEW_{k,\ell}) + 4k^2 + 2k\ell - 2k - (\ell + 4k^2 - 2k) \\
&= \OPT(\TRUEW_{k,\ell}) + 2\ETA_{k,\ell} - (\ell + 4k^2 - 2k) \\
&= \OPT(\TRUEW_{k,\ell}) + \left(2 - \frac{\ell + 4k^2 - 2k}{2k^2 + k\ell - k}\right)\ETA_{k,\ell}.
\end{align*}
For any $r < 2$, there exists $k,\ell \in \ZZ^+$ such that $2 - \frac{\ell + 4k^2 - 2k}{2k^2 + k\ell - k} \geqslant r$. 
\end{proof}

\section{Separation by Random Order Analysis}

We show that $\GFTP$ has a better random order ratio than $\FTP$, separating the two algorithms.

\begin{theorem}\label{thm:ror_ftp_full_analysis}
$\ROR_{\FTP}(\eps) = 1 + 2\eps$.
\end{theorem}
\begin{proof}
Since $\FTP$ does not make online changes to $T_{\FTP}$, the competitive analysis of $\FTP$ translates directly to a random order analysis of $\FTP$. Hence, the result follows from Theorems~\ref{thm:cr_ftp_upper_bound} and~\ref{thm:cr_ftp_lower_bound}.
\end{proof}

We start with the following lower bound on \GFTP.

\begin{theorem}\label{thm:ror_cftp_lower_bound}
$\ROR_{\GFTP}(\eps) \geqslant 1 + \eps$.
\end{theorem}
\begin{proof}
Fix $k > 1$. For each $\ell \in \ZZ^+$, construct an WMST-instance as in Theorem~\ref{thm:cr_ftp_lower_bound}, $(G_{k,\ell},\PREDW_{k,\ell},\TRUEW_{k,\ell})$,  and modify it into $(G_{k,\delta,\ell},\PREDW_{k,\delta,\ell},\TRUEW_{k,\delta,\ell})$ by setting $\TRUEW_{k,\ell}((u,v)) = \PREDW_{k,\ell}((u,v)) = \delta$, where $0 < \delta < 1$. In this way, $\GFTP$ always accepts $(u,v)$, regardless of the order in which the true weights of the edges arrive. By construction of $G_{k,\delta,\ell}$, $\GFTP$ will, for each $i=1,2,\ldots,\ell$, have to accept either $(z_i,u)$ or $(z_i,v)$. As in Theorem~\ref{thm:cr_ftp_lower_bound}, $\ADV$ sets the true weight of all edges in $E_{T_{\FTP}}$ to be $2k+1$ and all edges in $E_{T_{\OPT}}$ to be $1$. Assume, without loss of generality, that $E_{T_{\FTP}}$ contains all edges $(z_i,v)$ such that $\TRUEW_{k,\delta,\ell}((z_i,v)) = 2k+1$ and $\TRUEW_{k,\delta,\ell}((z_i,u)) = 1$, for all $i=1,2,\ldots,\ell$. Then, $\GFTP$ only replaces $(z_i,v)$ with $(z_i,u)$ if $(z_i,u)$ is revealed before $(z_i,v)$. Since the edges arrive uniformly at random, for each $i=1,2,\ldots,\ell$, $(z_i,u)$ is revealed before $(z_i,v)$ with probability $\frac{1}{2}$. Recalling that $\eps=k$ and $\OPT(\TRUEW_{k,\delta,\ell})=\ell + \delta$,
\begin{align*}
\mathbb{E}_\sigma[\GFTP(\PREDW_{k,\delta,\ell},\TRUEW_{k,\delta,\ell})] &\geqslant \OPT(\TRUEW_{k,\delta,\ell}) + \tfrac{1}{2}\cdot\ell\cdot 2k \\
&= \OPT(\TRUEW_{k,\delta,\ell}) + \eps \OPT(\TRUEW_{k,\delta,\ell}) - \delta k,
\end{align*}
and so
\begin{align*}
\frac{\mathbb{E}_\sigma[\GFTP(\PREDW_{k,\delta,\ell},\TRUEW_{k,\delta,\ell})]}{\OPT(\TRUEW_{k,\delta,\ell})} = 1 + \left(1 - \frac{\delta}{\ell+\delta}\right) \eps.
\end{align*}	
For all $r<1$ and all $0<\delta<1$, there exist $\ell \in \ZZ^+$ so $1 - \frac{\delta}{\ell+\delta} \geqslant r$. 
\end{proof}

We now turn to proving an upper bound of $1 + (1+\ln(2))\eps \approx 1
+ 1.69 \, \eps$  on the
random order ratio of \GFTP (Theorem~\ref{thm:ror_alg4}).
To this end, we apply the following lemmas.

\begin{lemma}\label{lem:bijectivity_lemma}
Let $G$ be a graph, and let $T_1$ and $T_2$ be two spanning trees of $G$. Then, for any edge $e_1 \in T_1 \setminus T_2$, there exists an edge $e_2 \in T_2 \setminus T_1$ such that $e_2$ introduces a cycle into $T_1$ that contains $e_1$, and $e_1$ introduces a cycle into $T_2$ that contains $e_2$.
\end{lemma}
\begin{proof}
Let $e_1 \in T_1 \setminus T_2$ be any edge, and let $u_1$ and $v_1$ denote the endpoints of~$e_1$. Removing $e_1$ from $T_1$ leaves $T_1$ disconnected. Let $\mathcal{U}$ and $\mathcal{V}$ denote the two connected components of $T_1 \setminus \{e_1\}$ such that $u_1 \in \mathcal{U}$ and $v_1 \in \mathcal{V}$. Since $T_2$ is spanning, there exists a $(u_1,v_1)$-path in $T_2$, along which there exists an edge $e_2$ that connects $\mathcal{U}$ and $\mathcal{V}$. Let $u_2$ and $v_2$ denote the endpoints of $e_2$ such that $u_2 \in \mathcal{U}$ and $v_2 \in \mathcal{V}$. Since $T_1\setminus \{e_1\}$ contains no $(\mathcal{U},\mathcal{V})$-edges, $e_2 \not\in T_1$, and, thus, $e_2 \in T_2\setminus T_1$, and $e_1$ introduces a cycle into $T_2$ that contains $e_2$. Moreover, since $T_1$ is spanning and $\mathcal{U}$ and $\mathcal{V}$ are connected, $\mathcal{U}$ contains a $(u_1,u_2)$-path and $\mathcal{V}$ contains a $(v_1,v_2)$-path. Thus, $e_2$ also introduces a cycle into $T_1$ that contains $e_1$.
\end{proof}

\begin{lemma}
\label{lem:important_detail}
Let $e \in T_{\FTP} \cap U$. If, at any point, an edge $e'$ introduces a cycle in $T$ that contains $e$, then $\PREDW(e) \leqslant \PREDW(e')$.
\end{lemma}
\begin{proof}
We prove this result by induction on the number of edges that $\GFTP$ has swapped out, i.e., by the number of edges $e'$ that have introduced a cycle $C$ into $T$, for which there existed an edge $e \in C \cap U$ such that $\TRUEW(e') \leqslant \PREDW(e)$.

\textbf{Base case:} Initially, $T = T_{\FTP}$. Now, let $e' \not\in T$ be any edge that introduces a cycle $C$ into $T$. In this case, by the minimality of $\OPT$, for any $e \in C$, it follows that $\PREDW(e) \leqslant \PREDW(e')$.

\textbf{Induction hypothesis:} Suppose that $\GFTP$ has swapped out $k-1$ edges. Let $e' \not\in T$ be any edge, and let $C$ be the cycle that $e'$ introduces into $T$. Then, for any $e \in C \cap U$, we have that $\PREDW(e) \leqslant \PREDW(e')$. 

\textbf{Induction step:} Suppose that $\GFTP$ has swapped out $k$ edges from $T$. Denote by $e_k'$ the last edge which introduced a cycle $C_k$ into $T$ that made $\GFTP$ swap out an edge $e_k$ for $e_k'$. Also, let $e' \not\in T$ be any edge, and denote by $C$ the cycle that $e'$ introduces into $T$. Now, let $e \in C \cap U$ be any unseen edge in $C$. We show that $\PREDW(e) \leqslant \PREDW(e')$. 

To this end, note that since $e_k'$ introduced a cycle into $T$ that contained $e_k$, we find that $e_k'$ and $e_k$ are two alternative edges that connect the same two connected components in $T$. Hence, if $e_k' \not\in C$, then $C$ is contained in one of these components, implying that before swapping out $e_k$ for $e_k'$, $e'$ introduced the same cycle $C$ into $T$. The result follows by the induction hypothesis. 

On the other hand, if $e_k' \in C$, then, before swapping out $e_k$ for $e_k'$, $e'$ would have introduced another cycle in $T$. We finish the analysis conditioned on whether (a) $e$ is not contained in $C_k$ or (b) $e$ is contained in $C_k$.

In case (a), we find that before swapping out $e_k$ for $e_k'$, $T$ contained a 	path that connected the endpoints of $e_k'$, without using $e$. Hence, before the swap, $e'$ would have introduced a cycle into $T$ that contained $e$, obtained by following $C$, except that we use the above path that connects the endpoints of $e_k'$ rather than using $e_k'$. Then, the induction hypothesis applies.

In case (b), since $T$ is a tree, we have the following picture:
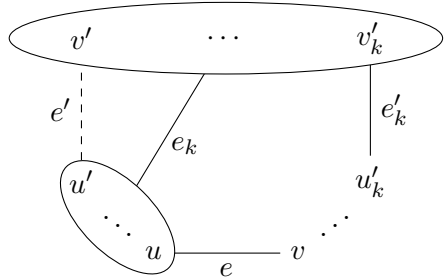
\begin{center}
\begin{tikzpicture}[scale = 0.95]
%\draw[gray!25] (0,0) grid (7,5);

\node (u) at (1,0) {$u$};
\node (v) at (3,0) {$v$};
\node (u') at (0,1) {$u'$};
\node (v') at (0,3) {$v'$};
\node (uk') at (4,1) {$u_k'$};
\node (vk') at (4,3) {$v_k'$};

\draw (u) -- node[below] {$e$} (v);
\draw[dashed] (u') -- node[left] {$e'$} (v');
\draw (uk') -- node[right] {$e_k'$} (vk');
\draw (2,3) -- node[right, pos = 0.6] {$e_k$} (0.5,0.5);

\draw[fill = white] (2,3) ellipse (3cm and 0.5cm);
\draw[rotate around = {-45:(0.5,0.5)}, fill = white] (0.5,0.5) ellipse (1cm and 0.5cm);

\node (u) at (1,0) {$u$};
\node (u') at (0,1) {$u'$};
\node (v') at (0,3) {$v'$};
\node (vk') at (4,3) {$v_k'$};

\node at (0.5,0.5) {$\ddots$};
\node at (3.5,0.5) {$\iddots$};
\node at (2,3) {$\cdots$};
\end{tikzpicture}
\end{center}
In this case, before swapping out $e_k$ for $e_k'$, we have that $e_k'$ introduced a cycle into $T$ that contained both $e_k$ and $e$, both of which were unseen at this point. Since $\GFTP$ always evicts the heaviest predicted edge in case of a swap, it follows that $\PREDW(e) \leqslant \PREDW(e_k)$. Now, if $e' = e_k$, then we are done. If, on the other hand $e' \neq e_k$, we find that, before swapping out $e_k$ for $e_k'$, $e'$ would have introduced a cycle into $T$ that contained $e_k$. Hence, by the induction hypothesis, $\PREDW(e_k) \leqslant \PREDW(e')$, and so $\PREDW(e) \leqslant \PREDW(e')$.
\end{proof}

\begin{lemma}\label{lem:expected_cost_function}
For all integers $n \geqslant 2$, we have that \[\frac{1}{n-1}\sum_{i=0}^{n-2} \left( 1 +  \frac{n-1}{2n-2-i} \right) \leqslant 1 + \ln(2)\,.\]
\end{lemma}
\begin{proof}
Let $f(n) = \frac{1}{n-1}\sum_{i=0}^{n-2} \left( 1 +  \frac{n-1}{2n-2-i} \right)$. It is sufficient to show that for all integers $n \geq 2$, we have that
\begin{enumerate}[label = {(\roman*)}]
\item $f(n) < f(n+1)$, and
\item $\lim_{n\to\infty} f(n) = 1 + \ln(2)$
\end{enumerate}
Towards (i), 
\begin{align*}
f(n) &= \frac{1}{n-1} \sum_{i=0}^{n-2} \left(1 + \frac{n-1}{2n-2-i} \right) = 1 + \sum_{i=0}^{n-2} \frac{1}{2n-2-i}. 
\end{align*}
Hence, for all integers $n \geq 2$,
\begin{align*}
f(n+1) - f(n) & = \sum_{i=0}^{n-1} \frac{1}{2n-i} - \sum_{i=0}^{n-2} \frac{1}{2n-2-i} = 
\frac{1}{2n} + \frac{1}{2n-1} - \frac{1}{n}\\
& =\frac{1}{2n(2n-1)} > 0,
\end{align*}
and so (i) follows. Towards (ii),
\begin{align*}
\lim_{n\to\infty} f(n) &=  1 + \lim_{n\to\infty} \sum_{i=0}^{n-2} \frac{1}{2n-2-i} = 1 + \lim_{n\to\infty} \left( \left(\sum_{i=1}^{2n-2} \frac{1}{i} \right)  - \left( \sum_{i=1}^{n-1} \frac{1}{i} \right)  \right) \\
& = 1 + \lim_{n\to\infty} \left( H_{2n-2} - H_{n-1} \right) \leqslant 1 +  \ln(2)\,,
\end{align*}
where $H_n = \sum_{i=1}^n \frac{1}{i}$ is the $n$th Harmonic number.
The last inequality follows, since  $\lim_{n\to\infty}\left( H_n - \ln(n)\right) = \gamma$ (where $\gamma$ is Euler's constant), and hence,
\begin{align*}
&\lim_{n\to\infty}\left( H_{2n-2} - \ln(2n-2)\right) - \lim_{n\to\infty} \left( H_{n-1} - \ln(n-1)\right) = \gamma - \gamma = 0 && \Leftrightarrow\\
&\lim_{n\to\infty} \left( H_{2n-2} - H_{n-1} \right) - \ln(2)=  0,
\end{align*}
and so, $\lim_{n\to\infty} \left( H_{2n-2} - H_{n-1} \right)  = \ln(2)$. 
\end{proof}

\begin{lemma}\label{lem:worst_case}
Suppose that $\GFTP$ has just rejected $e' \not\in T$. Then, at any future point, any unseen edge $e$ that is contained in the cycle that $e'$ introduces into $T$, at that point, satisfies that $\PREDW(e) < \TRUEW(e')$. 
\end{lemma}
\begin{proof}
Since $e'$ has just been rejected by $\GFTP$, it follows that for all $e \in C \cap U$, we have that $\PREDW(e) < \TRUEW(e')$, where $C$ is the cycle that $e'$ creates in $T$. The only way $C$ can be changed is if $\GFTP$ makes a swap, swapping out an edge $e_1 \in C \cap U$ for an edge $e_1' \not\in T$ which introduced a cycle $C_1$ into $T$, satisfying that $e_1 = \text{arg}\,\max_{e_i \in C_1 \cap U}\{\PREDW(e_i)\}$. After swapping out $e_1$ for $e_1'$, the cycle that $e'$ now introduces into $T$ is obtained by following $C$ and then the path connecting the endpoints of $e_1$, induced by $C_1$. In this case, for each $e_1^u \in C_1 \cap U$, we have that $\PREDW(e_1^u) \leqslant \PREDW(e_1)$. Since $\PREDW(e_1) < \TRUEW(e')$, we get that $\PREDW(e_1^u) < \TRUEW(e')$. This argument may be repeated if $\GFTP$ makes further changes to $C$. 
\end{proof}

\begin{theorem}\label{thm:ror_alg4}
$\ROR_{\GFTP}(\eps) \leqslant 1 + (1+\ln(2))\eps$.
\end{theorem}
\begin{proof}
  Given a WMST-instance $\WMSTI$, we let $T_{\GFTP,\sigma}$ denote the output tree that $\GFTP$ constructs when run on $\WMSTI$, where the order in which the weights arrive has been permuted according to a uniformly randomly chosen permutation $\sigma$ of $\{1,2,\ldots,m\}$. Further, we denote by $\GFTP(\PREDW,\TRUEW,\sigma)$ the cost of $T_{\GFTP,\sigma}$.

  The idea towards a random order ratio upper bound for $\GFTP$ is to prove the existence of a subset $E_{\BLAME} \subset T_{\OPT} \cup T_{\GFTP,\sigma}$ such that
\begin{align*}
\mathbb{E}_\sigma[\GFTP(\PREDPERMW,\TRUEPERMW,\sigma)] - \OPT(\TRUEPERMW) \leqslant \sum_{e \in E_{\BLAME}} \abs{\PREDW(e) - \TRUEW(e)}.
\end{align*}
and
\[\mathbb{E}_\sigma[\abs{E_{\BLAME}}] \leqslant n-1 + (n-1)\ln(2)\,.\]

More specifically, we define a function $\BETA \colon T_{\GFTP,\sigma} \rightarrow T_{\OPT}$ and prove that $\BETA$ is bijective, implying that
\begin{align*}
\mathbb{E}_\sigma[\GFTP(\PREDPERMW,\TRUEPERMW,\sigma)] - \OPT(\TRUEPERMW) = \sum_{e \in T_{\GFTP,\sigma}} \left( \TRUEW(e) - \TRUEW(\BETA(e)) \right),
\end{align*}
and then, for each $e \in T_{\GFTP,\sigma}$, argue that $\TRUEW(e) -\TRUEW(\BETA(e))$
is upper bounded by the prediction error of either
$e$ or $\BETA(e)$, or the sum of the two.
Then, we show that the expected number of edges for which the
upper bound is the error of both $e$ and $\beta(e)$ is upper bounded by $(n-1)\ln(2)$.
We also show that the edges whose errors are used as upper bounds are all distinct.

For the remainder of this proof, we denote by $T'$ a spanning tree of $G$ that is initially set to $T_{\OPT}$, and which we use to construct $\BETA$. Moreover, we denote by $X$ the random variable which is the size of $E_{\BLAME}$. Finally, $i$ is a random variable that counts the number of edges that have either been accepted by $\GFTP$ (now in $T'\cap T\cap \overline{U}$), or belong to $T' \setminus T$ and have been rejected (now in $(T'\setminus T)\cap\overline{U}$). It will be clear that $i$ counts the number of times $X$ increases (by either $0$, $1$ or $2$, the number of edges blamed). Since $\GFTP$ has to accept exactly $n-1$ edges, $X$ has to increase $n-1$ times, and so we can upper bound the expected cost difference between $T_{\GFTP,\sigma}$ and $T_{\OPT}$ when $i=n-1$.

We use $T'$ to keep track of which edges in $T_{\OPT}$ have been associated with an edge in $T_{\GFTP,\sigma}$ under $\BETA$. Any time $\GFTP$ accepts an edge $e$, we associate $e$ with an edge $e' \in T'$ under $\BETA$. We consider two cases:
\begin{enumerate}[label = {(\alph*)}]
\item \label{enum:a} If $e \in T'$, we set $\BETA(e) = e$, and so $T'$ remains unchanged.
\item \label{enum:b} If $e \not\in T'$, then Lemma~\ref{lem:bijectivity_lemma} implies that there exists an edge $e' \in T' \setminus T$ such that $e'$ introduces a cycle into $T$ that contains $e$, and $e$ introduces a cycle into $T'$ that contains $e'$. We select such an edge $e'$, set $\BETA(e) = e'$, and replace $e'$ by $e$ in $T'$.
\end{enumerate}
We repeat this process every time $\GFTP$ accepts an edge. This, however, requires $T'$ to remain a spanning tree at all times. To see that $T'$ remains a spanning tree, we note that in case~\ref{enum:a}, $T'$ remains unchanged and is therefore still a spanning tree. In case~\ref{enum:b}, we replace $e'$ with $e$ in $T'$. Since $e$ introduces a cycle into $T'$ that contains $e'$, it follows that $T'$ remains acyclic after the replacement, and so $T'$ is still spanning.

\textbf{Towards bijectivity of $\bm{{\BETA}}$:} In case~\ref{enum:a}, $\BETA(e) = e$, and so $e \in (T \cap T' \cap \overline{U})$. Since $e \not\in (T \cap T' \cap U) \cup (T' \setminus T)$, we never map to $e$ again later. In case~\ref{enum:b}, $\BETA(e) = e'$, and after replacing $e'$ by $e$ in $T'$, we find that $e \in (T \cap T' \cap \overline{U})$, and $e' \in \overline{T \cup T'}$. Hence, as neither $e$ nor $e'$ is contained in $(T \cap T' \cap U) \cup (T' \setminus T)$, we never map to either again later. Hence $\BETA$ is injective, and since $\abs{T_{\GFTP,\sigma}} = \abs{T_{\OPT}}$, $\BETA$ is bijective. 

\textbf{Invariant:} We present some structural observations, and prove they are true at any time while $\GFTP$ processes the input sequence permuted by~$\sigma$. 
\begin{enumerate}[label = {(\roman*)}]
\item \label{enum:T_minus_T'_unseen} Any edge in $T \setminus T'$ is unseen.
\item \label{enum:prob} For any $0 \leqslant i \leqslant n-2$, the probability that the next edge is contained in $T \setminus T'$, denoted $p_i$, satisfies
\begin{align*}
p_i = \frac{\abs{(T \setminus T') \cap U}}{\abs{E(G) \cap U}} \leqslant \frac{n-1}{2n-2-i},
\end{align*}
\item \label{enum:bound} For each edge $e \in T_{\GFTP,\sigma}$,
\begin{enumerate}[label = {(\arabic*)}]
\item \label{enum:edge_in_T} if $e$ was accepted in Line~\ref{alg-in-accept} of Algorithm~\ref{alg:gftp}, then
\begin{align*}
\TRUEW(e) - \TRUEW(\BETA(e)) \leqslant \abs{\TRUEW(e) - \PREDW(e)} + \abs{\TRUEW(\BETA(e)) - \PREDW(\BETA(e))}.
\end{align*}
\item \label{enum:swap} if $e$ was accepted after a swap in Line~\ref{alg-out-accept} of Algorithm~\ref{alg:gftp}, then
\begin{align*}
\TRUEW(e) - \TRUEW(\BETA(e)) \leqslant \abs{\TRUEW(\BETA(e)) - \PREDW(\BETA(e))}.
\end{align*}
\end{enumerate}
\end{enumerate}

\textbf{Towards~\ref{enum:T_minus_T'_unseen}:} Initially, all edges are unseen. If an edge $e$ in $T \setminus T'$ is revealed, we replace $\beta(e)$ with $e$ in $T'$, so now $e \in T \cap T'$. Hence, after replacing $e'$ with $e$ in $T'$, all edges in $T \setminus T'$ are again unseen. 

\textbf{Towards~\ref{enum:prob}:} Initially, we note that
\begin{align*}
p_i = \frac{\abs{(T \setminus T') \cap U}}{\abs{E(G) \cap U}} \leqslant \frac{\abs{(T \setminus T') \cap U}}{\abs{(T \cup T') \cap U}} 
\end{align*}
From~\ref{enum:T_minus_T'_unseen}, it follows that $\abs{(T \setminus T') \cap U} = n-1 - a_i - x_i$, where $a_i$ is the number of edges that have been accepted, i.e., the number of edges in $T \cap T' \cap \overline{U}$, and $x_i$ is the number of edges in $T \cap T' \cap U$. Then,
\begin{align*}
p_i \leqslant \frac{n-1 - x_i - a_i}{\abs{(T \cup T') \cap U}}.
\end{align*}
Now,
\begin{align*}
(T \cup T') \cap U &= (T \cup T')\setminus ((T \cup T') \cap \overline{U}) \\
\Rightarrow \hspace{0.5cm} \abs{(T \cup T') \cap U} &= \abs{T \cup T'} - \abs{(T \cup T') \cap \overline{U}}.
\end{align*}
For any $0 \leqslant i \leqslant n-2$,
\begin{align*}
\abs{T \cup T'} &= \abs{T} + \abs{T'} - \abs{T \cap T'} = 2n-2 - x_i - a_i, \\
\text{and} \hspace{0.5cm} \abs{(T \cup T') \cap \overline{U}} &= \abs{(T \setminus T') \cap \overline{U}} + \abs{T' \cap \overline{U}} = \abs{T' \cap \overline{U}} =i.
\end{align*}
Here the second to last equality follows from~\ref{enum:T_minus_T'_unseen}, and the last equality follows from the definition of $i$. Hence,
\begin{align*}
p_i \leqslant \frac{n-1 - x_i - a_i}{2n-2-i - x_i-a_i}.
\end{align*}
Using that $a_i+x_i \geqslant 0$ and $i \leqslant n-1$, it follows that
\begin{align*}
p_i \leqslant \frac{n-1}{2n-2-i}.
\end{align*}

\textbf{Towards~\ref{enum:bound}:}
At any point, before the true weight of the next edge, $\ENEXT$, is revealed, we may decompose $E(G)$ into the following disjoint union:
\begin{align*}
E(G) = (((T \cap T') \cup (\overline{T \cup T'}) \cup (T \setminus T') \cup (T' \setminus T)) \cap U) \cup \overline{U},
\end{align*}
where $U$ is the collection of unseen edges. We split the analysis into cases based on which set $\ENEXT$ is contained in. 
\begin{enumerate}[label = {},, labelwidth=\parindent , leftmargin=\parindent ,align=left]
\item \textbf{Case $\bm{{\ENEXT \in T \cap T'}}$:} $\GFTP$ accepts $\ENEXT$, and so, by~\ref{enum:a}, $\BETA(\ENEXT) = \ENEXT$, implying that $\TRUEW(\ENEXT) - \TRUEW(\BETA(\ENEXT)) = 0$, so that~\ref{enum:edge_in_T} is satisfied. Hence, $X$ increases by $0$, as no edges are added to $E_{\BLAME}$. In this case, we make no changes to $T$ nor $T'$, and we increment $i$.
\item \textbf{Case $\bm{{\ENEXT \in \overline{T \cup T'}}}$:} We split the analysis into two subcases.
\begin{enumerate}[label = {}]
\item \textbf{Subcase (accept):} If $\GFTP$ accepts $\ENEXT \in \overline{T \cup T'}$, it does so due to swapping out some edge $e \in T$ that is contained in the cycle that $\ENEXT$ introduces in $T$ by $\ENEXT$. In this case, by~\ref{enum:b}, there exists an edge $\EOPT \in T' \setminus T$ such that $\BETA(\ENEXT) =\EOPT$. We now argue that 
\begin{align*}
\TRUEW(\ENEXT) - \TRUEW(\EOPT) \leqslant \abs{\PREDW(\EOPT) - \TRUEW(\EOPT)}.
\end{align*} 
Note that since $\GFTP$ swapped out $e$ for $\ENEXT$, we have that $\TRUEW(\ENEXT) \leqslant \PREDW(e)$. Further, we can argue that $\PREDW(e) \leqslant \PREDW(\EOPT)$. Indeed, if $e = \EOPT$, this is trivial. If $e \neq \EOPT$, then, since $\EOPT$ introduces a cycle that contains $\ENEXT$, it follows that before swapping out $e$ for $\ENEXT$, $\EOPT$ would introduce a cycle in $T$ containing $e$, and so $\PREDW(e) \leqslant \PREDW(\EOPT)$, by Lemma~\ref{lem:important_detail}. Hence,
\begin{align*}
\TRUEW(\ENEXT) - \TRUEW(\EOPT) &\leqslant \PREDW(e) - \TRUEW(\EOPT) \\
&\leqslant \PREDW(\EOPT) - \TRUEW(\EOPT) \\
&\leqslant \abs{\PREDW(\EOPT) - \TRUEW(\EOPT)},
\end{align*}
so~\ref{enum:swap} is satisfied. In this case, $X$ increases by $1$, as we add $\EOPT$ to $E_{\BLAME}$, we replace $\EOPT$ with $\ENEXT$ in $T'$, and we increment $i$.
\item \textbf{Subcase (reject):} If $\GFTP$ rejects $\ENEXT \in \overline{T \cup T'}$, we do not associate any edges, we make no further changes, $X$ is unchanged as no edges have been accepted, and we do not increment $i$.
\end{enumerate}
\item \textbf{Case $\bm{{\ENEXT \in T \setminus T'}}$:} In this case, $\GFTP$ accepts $\ENEXT$. By~\ref{enum:b}, there exists an edge $\EOPT \in T' \setminus T$ such that $\BETA(\ENEXT) = \EOPT$. Since $T$ remains unchanged when $\ENEXT$ is revealed, it follows that $\EOPT$ would introduce a cycle in $T$ containing $\ENEXT$ before $\ENEXT$ was revealed. Hence, by Lemma~\ref{lem:important_detail}, we find that $\PREDW(\ENEXT) \leqslant \PREDW(\EOPT)$, and so
\begin{align*}
\TRUEW(\ENEXT) - \TRUEW(\EOPT) &= \TRUEW(\ENEXT) - \PREDW(\ENEXT) + \PREDW(\ENEXT) - \TRUEW(\EOPT) \\
&\leqslant \TRUEW(\ENEXT) - \PREDW(\ENEXT) + \PREDW(\EOPT) - \TRUEW(\EOPT) \\
&\leqslant \abs{\TRUEW(\ENEXT) - \PREDW(\ENEXT)} + \abs{\PREDW(\EOPT) - \TRUEW(\EOPT)},
\end{align*}
so~\ref{enum:edge_in_T} is satisfied. Then, $X$ is increased by $2$ as we add both $\ENEXT$ and $\EOPT$ to $E_{\BLAME}$, we replace $\EOPT$ by $\ENEXT$ in $T'$, and increment $i$.
\item \textbf{Case $\bm{{\ENEXT \in T' \setminus T}}$:} In this case, $\ENEXT$ introduces a cycle $C$ in $T$. Denote by $e$ an edge in $C \cap U$ for which $e=\text{arg}\,\max_{e_i \in C \cap U}\{\PREDW(e_i)\}$. We split the remaining analysis into two subcases.
\begin{enumerate}[label = {}]
\item \textbf{Subcase (accept):} If $\TRUEW(\ENEXT) \leqslant \PREDW(e)$, then $\GFTP$ accepts $\ENEXT$ and removes $e$ from its tree. Then, by~\ref{enum:a}, $\BETA(\ENEXT) = \ENEXT$ and so $\TRUEW(\ENEXT) - \TRUEW(\BETA(\ENEXT)) = 0$, so that~\ref{enum:swap} is satisfied. Then, $X$ increases by $0$, as no edges are added to $E_{\BLAME}$. We make no further changes to $T$ or $T'$, and we increment $i$. 
\item \textbf{Subcase (reject):} If $\PREDW(e) < \TRUEW(\ENEXT)$, then $\GFTP$ rejects $\ENEXT$. Since each edge in $T'$ will, at some point, be associated with an edge in $T_{\GFTP,\sigma}$, by the bijectivity of $\BETA$, it follows that $\GFTP$ will later accept some edge that will be associated with $\ENEXT$ under $\BETA$. Denote this edge by $\EFUTURE$, such that $\BETA(\EFUTURE) = \ENEXT$. Note that $\GFTP$ can accept $\EFUTURE$ either due to a swap, or because $\EFUTURE$ was revealed while contained in $T$. If $\EFUTURE$ is accepted due to a swap, then, by the above, we can upper bound any extra incurred cost by the prediction error of $\BETA(\EFUTURE) = \ENEXT$, and so we add $\ENEXT$ to $E_{\BLAME}$. On the other hand, suppose that $\EFUTURE$ is being accepted as its true weight is revealed while $\EFUTURE$ is contained in $T$. In this case, at the time where $\GFTP$ accepts $\EFUTURE$, we find that $\EFUTURE$ is contained in the cycle that $\ENEXT$ introduces into $T$, and so, by Lemma~\ref{lem:worst_case}, $\PREDW(\EFUTURE) < \TRUEW(\ENEXT)$, implying that
%\begin{align*}
$\TRUEW(\EFUTURE) - \TRUEW(\ENEXT) < \TRUEW(\EFUTURE) - \PREDW(\EFUTURE)$.
%\end{align*}
Then, we add $\EFUTURE$ to $E_{\BLAME}$. In either case, $X$ increases by $1$, we make no changes to $T$ or $T'$, and we increment $i$.
\end{enumerate}
\end{enumerate}
Having established the invariant,
the only time $X$ increases by $2$ is if $\ENEXT \in T \setminus T'$. Now, by~\ref{enum:prob}, for each $i$, this happens with probability at most $\frac{n-1}{2n-2-i}$. In any other case, we add at most $1$ to $X$. Hence, $X$ satisfies that
\begin{align*}
\mathbb{E}_\sigma[X] \leqslant \sum_{i=0}^{n-2} \left( 2 p_i + (1-p_i) \right) = \sum_{i=0}^{n-2} \left( 1 + \frac{n-1}{2n-2-i} \right) \leqslant (n-1)(1 + \ln(2)),
\end{align*}
where the last inequality follows from Lemma~\ref{lem:expected_cost_function}.

In the following, we argue that we never use the prediction error of an edge to upper bound incurred cost more than once. To this end, let $e$ be an edge that $\GFTP$ has just accepted. Then, by \ref{enum:bound} we can upper bound $\TRUEW(e) - \TRUEW(\BETA(e))$ by either the prediction error of $e$ or $\BETA(e)$, or the sum of the two. By the proof of bijectivity of $\BETA$, it follows that $e$ can never be hit under $\BETA$, and so we will never consider using the prediction error of $e$ again later. On the other hand, after replacing $\BETA(e)$ with $e$ in $T'$, we have that $\BETA(e) \in \overline{T \cup T'}$. As $\BETA(e)$ may still be unseen, it follows that $\GFTP$ may accept $\BETA(e)$ later due to a swap. In this case, by~\ref{enum:swap}, it follows that $w(\BETA(e)) - w(\BETA(\BETA(e))$ can be upper bounded by the prediction error of $\BETA(\BETA(e))$, and so, we never use the prediction error of $\BETA(e)$ later to upper bound extra incurred cost.

Since we only use the prediction error of an edge to upper bound incurred cost once, and since the largest $n-1$ prediction errors upper bound the concrete prediction errors used, it follows that
\begin{align*}
\mathbb{E}_\sigma[\GFTP(\PREDPERMW,\TRUEPERMW,\sigma) - \OPT(\TRUEPERMW)] \leqslant (1 + \ln(2))\ETA,
\end{align*}
so
\begin{align*}
\frac{\mathbb{E}_\sigma[\GFTP(\PREDPERMW,\TRUEPERMW,\sigma)]}{\OPT(\TRUEPERMW)} \leqslant 1 + (1+\ln(2))\eps,
\end{align*}
and, hence,
%\begin{align*}
$\ROR_{\GFTP}(\eps) \leqslant 1 + (1+\ln(2))\eps$.
%\end{align*}
\end{proof}

\section{Open Problems}

An obvious open problem is to determine the exact random order ratio of \GFTP,
in the range $1+\eps$ to $1+\ln(2)\eps$.

\GFTP can be seen as an improvement of \FTP, and we are interested in
what we believe could be a further improvement: In addition to
accepting some edges that are not in the chosen minimum spanning tree
based on predictions, also reject \emph{some} that \emph{are} in that
tree, if the actual weight is higher than the predicted. The
obvious approach gives an algorithm with a worse competitive
ratio than \FTP's, but restricting which edges the algorithm can
accept after such a rejection gives rise to another optimal algorithm
under competitive analysis.  It would be interesting to apply random
order analysis to such an algorithm as well.

More generically, it would be interesting to apply
random order analysis to other online problems with predictions,
as well as to consider error measures similar to ours for other problems.

\bibliography{refs.bib}
\bibliographystyle{plain}

\end{document}